\documentclass[submission]{eptcs} 
\providecommand{\event}{Gandalf 2021} 
\usepackage{breakurl}             
\usepackage{underscore}           

\title{On the size of disjunctive formulas in the $\mu$-calculus}
\author{Clemens Kupke\thanks{Partially supported by Leverhulme grant RPG-2020-232.}
\institute{University of Strathclyde\\
Glasgow, Scotland}
\email{clemens.kupke@strath.ac.uk}
\and
Johannes Marti\thanks{The research of this author has been made possible by a grant from the Dutch Research Council NWO, project nr. 617.001.857.}
\quad\quad Yde Venema
\institute{ILLC, University of Amsterdam\\ 
Amsterdam, The Netherlands}
\email{johannes.marti@gmail.com \quad\quad y.venema@uva.nl}
}
\def\titlerunning{Size of Disjunctive Formulas}
\def\authorrunning{C.~Kupke, J.~Marti \& Y.~Venema}

\usepackage{amssymb}
\usepackage{amsmath}
\usepackage{enumitem}
\usepackage{latexsym}
\usepackage{stmaryrd}
\usepackage{tikz-qtree}
  \usetikzlibrary{automata,positioning}

\newtheorem{theorem}{Theorem}[section]

\newtheorem{proposition}[theorem]{Proposition}

\newtheorem{corollary}[theorem]{Corollary}

\newtheorem{definition}[theorem]{Definition}

\newenvironment{proof}{\begin{trivlist}\item[\hskip\labelsep{\bf
Proof.\ }]}{\hspace*{\fill} {\sc qed}\end{trivlist}}
\newenvironment{proofof}[1]{\begin{trivlist}\item[\hskip\labelsep{\bf
Proof~of~{#1}.\ }]}{\hspace*{\fill} {\sc qed}\end{trivlist}}

\newenvironment{tbs}{%
   \small\tt
   \begin{itemize}}{\end{itemize}}
\newcommand{\btbs}{\begin{tbs}}                                                                      
\newcommand{\etbs}{\end{tbs}}                                                                      

\newcommand{\mathstr}[1]{\mathbb{#1}}

\newcommand{\bbA}{\mathstr{A}}

\newcommand{\bbG}{\mathstr{G}}
\newcommand{\bbH}{\mathstr{H}}

\newcommand{\bbP}{\mathstr{P}}

\newcommand{\bbS}{\mathstr{S}}

\newcommand{\bbW}{\mathstr{W}}

\newcommand{\Prop}{\ensuremath{\mathsf{P}}}        

\newcommand{\ML}{\ensuremath{\mathtt{ML}}}

\newcommand{\DML}{\ensuremath{\mathrm{DML}}}

\newcommand{\Lit}{\mathtt{Lit}}
\newcommand{\At}{\mathtt{At}}

\newcommand{\MLone}{\ensuremath{\mathtt{1ML}}} 
\newcommand{\DMLone}{\ensuremath{\mathtt{1DML}}} 

\newcommand{\AMLone}{\mathtt{1AML}}

\newcommand{\muML}{\ensuremath{\mu\ML}}    
\newcommand{\muDML}{\ensuremath{\mu\DML}}    

\newcommand{\lneg}[1]{\ol{#1}}

\newcommand{\bv}{\bigvee}

\newcommand{\dia}{\Diamond}

\newcommand{\nb}{\nabla}

\newcommand{\hs}{\heartsuit}

\newcommand{\Sfor}{\ensuremath{\mathit{Sfor}}}

\newcommand{\asz}[1]{|#1|^{s}} 

\renewcommand{\phi}{\varphi} 
\newcommand{\isbnf}{\mathrel{::=}}
\newcommand{\divbnf}{\;\mid\;}

\newcommand{\V}{\mathrm{Val}}

   \newcommand{\pow}{\wp}

\newcommand{\Acc}{\mathit{Acc}}
\newcommand{\Lom}{L_{\om}}

\newcommand{\NBT}{\mathrm{NBT}}

\newcommand{\idx}{\mathit{ind}}

\newcommand{\AG}{\mathcal{A}}

\newcommand{\EG}{\mathcal{E}}

\newcommand{\eloi}{\exists}
\newcommand{\abel}{\forall}

\newcommand{\Win}{\mathrm{Win}}

\newcommand{\LS}{\mathit{LS}}

\newcommand{\nada}{\varnothing}

\newcommand{\sse}{\subseteq}

\newcommand{\Diag}{\Delta}

\newcommand{\size}[1]{|#1|}
\newcommand{\first}{\mathit{first}}
\newcommand{\last}{\mathit{last}}

\newcommand{\isdef}{\mathrel{:=}}

\newcommand{\Dom}{\mathsf{Dom}}
\newcommand{\Ran}{\mathsf{Ran}}

\newcommand{\parto}{\stackrel{\circ}{\to}}

\newcommand{\De}{\Delta}
\newcommand{\Th}{\Theta}
\newcommand{\Si}{\Sigma}
\newcommand{\Om}{\Omega}
\newcommand{\al}{\alpha}

\newcommand{\si}{\sigma}
\newcommand{\om}{\omega}

\newcommand{\coloneqq}{\mathrel{:=}}

\newcommand{\ol}[1]{\overline{#1}}

\newcommand{\Ord}{\mathcal{O}}

\newcommand{\wt}[1]{\widetilde{#1}}

\newcommand{\comp}{\mathrel{;}}

\begin{document}
\maketitle

\begin{abstract}
A key result in the theory of the modal $\mu$-calculus is the
disjunctive normal form theorem by Janin \& Walukiewicz, stating that
every $\mu$-calculus formula is semantically equivalent to a so-called
disjunctive formula. These disjunctive formulas have good computational
properties and play a pivotal role in the theory of the modal
$\mu$-calculus. It is therefore an interesting question what the best normalisation procedure is for rewriting a formula into an equivalent disjunctive formula of minimal size.
The best constructions that are known from the literature are automata-theoretic in nature and consist of a guarded transformation, i.e., the constructing of an equivalent guarded alternating automaton from a $\mu$-calculus formula, followed by a Simulation Theorem stating that any such alternating automaton can be transformed into an equivalent non-deterministic one. Both of these transformations are exponential constructions, making the best normalisation procedure doubly exponential. Our key contribution presented here shows that the two parts of the normalisation procedure can be integrated, leading to a procedure that is single-exponential in the closure size of the formula. 
\end{abstract}

\section{Introduction}

The modal $\mu$-calculus~\cite{brad:moda06} is a general modal logic enriched with fixpoint operators that allow
to reason about the ongoing, possibly infinite behaviour of a system. 
The generality and complexity of the modal $\mu$-calculus calls for research into 
fragments of the logic.
On the one hand, this concerns fragments tailor-made for certain application 
domains such as the temporal logics LTL, CTL or dynamic logics~\cite{stir:moda01}. On the other hand, one focuses on fragments of the $\mu$-calculus that are either semantically or syntactically well-behaved and where a better understanding increases our knowledge about the full $\mu$-calculus. One key fragment of the latter kind is formed by the so-called disjunctive formulas~\cite{jani:auto95} . These are formulas, where the use of conjunctions is strictly limited to conjunctions of propositional atoms and the formula $\top$ (thought of as the empty conjunction). 
We state the formal definition here - the meaning of the $\nb$-operator will be discussed later.


\begin{definition}
The set $\muDML$ of disjunctive $\mu$-calculus formulas is given by the following
grammar:
\[ \phi \isbnf \bot \mid \top 
   \mid \bigwedge L \land \nb \Phi 
   \mid (\phi_1 \lor \phi_2) \mid \mu p. \varphi  \mid \nu p. \varphi
\]
where $L$ is a finite set of literals (i.e., propositional variables or their 
negations),
$\Phi$ is a finite set of formulas in $\muDML$, 
and $p$ is a propositional variable. 
Furthermore we require that in a formula $\eta p.\phi$ all occurrences of $p$ in
$\phi$ are positive, guarded and not in the context of a conjunction $p \land 
\psi$.
\end{definition}

While disjunctive formulas correspond to a proper syntactic fragment of the
$\mu$-calculus, it is a somewhat surprising fact that each formula of the 
$\mu$-calculus is semantically equivalent to a disjunctive one. 
This has many applications, e.g., satisfiability checking of a disjunctive 
formula can be carried out efficiently~\cite{jani:auto95} (being 
ExpTime-complete in for arbitrary formulas~\cite{emer:comp88,emer:tree91}) 
and disjunctive formulas facilitate the computation of uniform
interpolants~\cite{dago:logi00,dale:onmo06}. 
Furthermore, disjunctive formulas play a pivotal role in the completeness
theory of the modal $\mu$-calculus~\cite{walu:comp00, enqv:comp18}. 
Finally, disjunctive formulas also provide insights for characterising important
semantic fragments such as the continuous, additive and monotone fragments of
the modal $\mu$-calculus~\cite{font:mode18}.

Recipes to rewrite a given arbitrary $\mu$-calculus formula into an equivalent disjunctive one are well-known~\cite{jani:auto95}.
The size of the resulting disjunctive formula is crucial, in particular, in connection with satisfiability and 
uniform interpolation. 
The construction of a disjunctive formula usually proceeds in two stages: first a
given $\mu$-calculus formula is transformed into an equivalent, possibly 
alternating modal automaton, which is then transformed into an equivalent 
non-deterministic ``disjunctive'' modal automaton. 
The latter can be easily translated back into a disjunctive formula. We will argue that the outlined two-stage construction will inevitably lead to a double-exponential blow-up in the size of the formula. This is, because the first move from formulas to automata involves bringing the formula into a guarded format, i.e.,  a form where each fixpoint variable is in the scope of at least one modality. That guarding is problematic has been observed in the work by Bruse et~al.~\cite{brus:guar15} - we will argue that it necessarily entails an exponential blow-up in the size of the structures involved. 

This sets the stage for our main result: a procedure that directly turns an arbitrary, possibly unguarded formula in the modal $\mu$-calculus into an equivalent disjunctive automaton of exponential size. The latter can be turned easily into a disjunctive formula, which leads to our main theorem.

\begin{theorem}\label{thm:main}
For any $\mu$-calculus formula $\varphi$ we can construct an equivalent 
disjunctive $\mu$-calulus formula $\phi^d$ of size $2^{\Ord(n^2k\cdot \log(nk))}$
and alternation depth $\Ord(n\cdot k)$, where $n = \size{\varphi}$ and where $k$ 
is the alternation depth of $\varphi$.
\end{theorem}

In the above theorem, the size of a formula refers to the size of its Fisher-Ladner closure, which has been shown in~\cite{brus:guar15} to provide the tightest measure of formula size.
For a discussion and comparison of different size measures 
see~\cite{kupk:size20} where we also propose
so-called ``parity formulas'' as a versatile tool to study the complexity of formula constructions. Parity formulas are a graph-shaped variant of $\mu$-calculus formulas which is closely related to Wilke's alternating automata~\cite{wilk:alte01} and hierarchical equation systems~\cite{demr:temp16}. 
While we stated the above theorem with reference to standard formulas, we will work throughout the paper with 
parity formulas instead. At the same time we will explain why Thm.~\ref{thm:main} is a consequence of our work.  

The outline of our paper is thus as follows: we will first introduce the necessary terminology for parity formulas, modal automata and their respective disjunctive variants. After that, in Section~3, we will discuss why guarding a parity formula can lead to an exponential blow-up. We then demonstrate that turning an arbitrary formula into an equivalent modal automaton is at least as costly as guarding a formula which means that the earlier mentioned two-stage method of constructing a disjunctive formula will in general lead to a double-exponential blow-up. Section 5 contains the central result of this paper, a construction that turns any given parity formula into an equivalent disjunctive modal automaton. This will provide a proof of Thm.~\ref{thm:main}.

{\em Related Work.} 
In addition to the already mentioned papers we would like to highlight a few more
closely related lines of research. 
In spirit, our construction is related to the work by Friedmann \& Lange on 
tableaux for unguarded $\mu$-calculus formulas~\cite{frla13:deci} but the cited 
paper is not concerned with disjunctive normal forms.  
Similarly, our automata theoretic result could be obtainable from a more general 
result in~\cite{vard98:reas} - a key definition in that paper, however, appears
to make an implicit assumption on the names of fixpoint variables (``cleanness'') 
whereas our results from~\cite{kupk:size20} demonstrate that cleaning a formula 
can lead to an exponential blow-up in (closure) size. 
In addition, it is not clear how to extract a disjunctive formula from the purely 
automata-theoretic constructions in~\cite{vard98:reas}. Finally, the work by 
Lehtinen~\cite{leht:disj15} studies how the alternation depth of a formula 
relates to the alternation depth of an equivalent disjunctive formula. 
While it turns out that the difference in alternation depth can be arbitrarily
big, we note that this is not in conflict with the bound in our theorem, as we
refer to a particular disjunctive equivalent as opposed to an arbitrary one.

{\em Acknowledgements.}
We would like to thank the anonymous referees for valuable comments that helped to improve this paper.

\section{Preliminaries}

\subsection{The $\mu$-calculus and parity formulas}

We will now recall the standard syntax of the modal $\mu$-calculus and its 
reformulation in terms of parity formulas. 
It will be convenient for us to assume that $\mu$-calculus formulas are in
so-called negation normal form.
We assume an infinite set of propositional variables, and define a \emph{literal} 
to be either a propositional variable $p$ or its negation $\lneg{p}$.

\begin{definition}
\label{d:syntnnf}
The formulas of the modal $\mu$-calculus $\muML$ are given by the following
grammar:
\begin{eqnarray*}
\phi & \isbnf & 
   \ell
   \divbnf \bot \divbnf \top
   \divbnf (\phi_{1}\lor\phi_{2}) \divbnf (\phi_{1}\land\phi_{2}) 
   \divbnf
   \dia\phi \divbnf \Box\phi 
   \divbnf \mu p. \phi \divbnf \nu p. \phi,
\end{eqnarray*}
where $\ell$ is a literal, $p$ is a propositional variable, and  the formation of
the formulas $\mu p. \phi$ and $\nu p. \phi$ is subject to the constraint that 
$\phi$ is \emph{positive} in $p$, i.e., there are no occurrences of $\lneg{p}$
in $\phi$.
With $\size{\varphi}$ we denote the size of a formula measure in the number
of distinct formulas in its Fisher-Ladner closure. 
\end{definition}

%
%
%
%
%
We often restrict attention to formulas of which the free variables belong to 
some fixed finite set $\Prop$; these are interpreted over Kripke models over 
$\Prop$ (in the following referred to as models), i.e., triples $\bbS = (S,R,\V)$
where $S$ is a set of points, $R$ is a binary relation and $\V: \Prop \to \pow S$. 
We sometimes refer to the propositional type $c_{s} \isdef \{ p \in \Prop \mid s
\in \V(p) \}$ of $s \in S$ as the \emph{colour} of $s$.
A {\em pointed model} is a model $\bbS$ together with a designated point 
$s_I \in S$. 
Finally, it will be convenient to extend $\V$ to all literals by putting 
$\V(\lneg{p}) = S \setminus \V(p)$.

In this paper, we will not work with $\mu$-calculus formulas in their usual 
shape, but with formulas represented as graphs, so-called ``parity formulas''. 
Parity formulas will facilitate discussing the complexity of our constructions. 
In addition, the fact that parity formulas resemble automata will simplify our 
proofs, as key constructions in our paper turn formulas into automata and vice
versa. 
While parity formulas were introduced in~\cite{kupk:size20} they are closely 
related to alternating automata~\cite{wilk:alte01} and hierarchical equation 
systems, see for instance~\cite{demr:temp16}. 
A detailed discussion of the connections can be found 
in~\cite[Section~5]{kupk:size20}. 
Before giving the definition it will be useful to fix some terminology for 
directed graphs (which we will also apply to structures such as parity formulas 
that possess a directed graph structure). For binary relations $R \subseteq X \times X$
and $x \in X$ we will use the notation $R[x]$ to denote the set $\{x' \in X \mid (x,x') \in R \}$.

\begin{definition}
  Let $(V,E)$ be a directed graph. A path $\pi$ through $(V,E)$
  is a finite, non-empty sequence $\pi = v_0 \dots v_n \in V^*$ such that
  $v_{i+1} \in E [v_i]$ for all $i \in \{ 0,\dots,n-1\}$.
  We denote by $\first(\pi)$ and $\last(\pi)$ the first and last element
  of the path $\pi$, respectively. Concretely, for the above path 
  we have $\first(\pi) = v_0$ and $\last(\pi) = v_n$. 
  A path $\pi$ with $\first(\pi) = \last(\pi)$ is called a cycle if it consists of at least two nodes.  
\end{definition}

Given the set $\Prop$ of proposition letters, we let $\Lit(\Prop)$ and 
$\At(\Prop) \isdef \Lit(\Prop) \cup \{ \bot,\top \}$ denote the set of 
\emph{literals} and \emph{atomic formulas over} $\Prop$, respectively.

\begin{definition}
\label{d:pf}
Let $\Prop$ be a finite set of proposition letters.
A \emph{parity formula over $\Prop$} is a quintuple $\bbG = (V,E,L,\Om,v_{I})$, 
where
\begin{itemize}
 \item $(V,E)$ is a finite, directed graph, with $\size{E[v]} \leq 2$ for 
every vertex $v$;
 \item $L: V \to \At(\Prop) 
      \cup \{ \land, \lor, \dia, \Box, \epsilon \}$ is a labelling function;
 \item $\Om: V \parto \om$ is a partial map, the \emph{priority} map of $\bbG$; and 
 \item $v_{I}$ is a vertex in $V$, referred to as the \emph{initial} node of $\bbG$;
\end{itemize}
\noindent
such that 
\begin{enumerate}
 \item $\size{E[v]} = 0$ if $L(v) \in \At(\Prop)$, and 
   $\size{E[v]} = 1$ if $L(v) \in \{ \dia, \Box\} \cup \{ \epsilon \}$;
 \item every cycle of $(V,E)$ contains at least one node in $\Dom(\Om)$.
\end{enumerate}
A node $v \in V$ is called 
 \emph{atomic} if it is either constant or literal,
 \emph{boolean} if $L(v) \in \{ \land, \lor \}$, and
\emph{modal} if $L(v) \in \{ \dia, \Box \}$.
We denote by $V_a$, $V_b$ and $V_m$ the
collections of atomic, boolean and modal nodes, respectively.
The elements of $\Dom(\Om)$ will be called \emph{states}.
The \emph{size} of a parity formula  $\bbG = (V,E,L,\Om,v_{I})$ is defined as
its number of nodes: $\size{\bbG} \isdef \size{V}$.
\end{definition}

\begin{definition}
\label{d:pfgam}
Let $\bbS=(S,R,\V)$ be a Kripke model for a set $\Prop$ of proposition letters, 
and let $\bbG = (V,E,L,\Om,v_{I})$ be a parity formula over $\Prop$.
We define the \emph{evaluation game} $\EG(\bbG,\bbS)$ as the parity game 
$(G,E,\Om')$ of which 
the board consists of the set $V \times S$, 
the priority map $\Om': V \times S \to \om$ is given by
\[
\Om'(v,s) \isdef \left\{ \begin{array}{ll}
      \Om(v) & \text{if } v \in \Dom(\Om)
   \\ 0      & \text{otherwise},
\end{array}\right.
\]
and the game graph is given in Table~\ref{tb:2}. 
Here all possible game positions are listed in the left column, the owner of a position is either $\forall$ or $\exists$\footnote{Note that we do not need to assign a player to positions that admit a single 
move only.}
 and the set of possible moves is specified in the right column. 
As usual, finite plays of the game are lost by the player who owns the last 
position of the play from which no more move is possible (``the player who gets
stuck loses''). 
An infinite play is won by $\exists$ if the maximum priority occurring infinitely 
often along the play is even, and by $\forall$ if it is odd.

The parity formula $\bbG$ \emph{holds} at a point $s$ if the pair $(v_{I},s)$ is 
winning for $\exists$ in the evaluation game.
\end{definition}

\begin{table}[t]
\begin{center}
\begin{tabular}{|ll|c|c|}
\hline
\multicolumn{2}{|l|}{Position} & Player  & Admissible moves 
\\\hline
     $(v,s)$ & with $L(v) = l$ and $s \in \V(l)$         
   & $\abel$ & $\nada$ 
\\   $(v,s)$ & with $L(v) = l$ and $s \notin \V(l)$      
   & $\eloi$ & $\nada$ 
\\   $(v,s)$ & with $L(v) = \epsilon$ 
   & - & $E[v] \times \{ s \}$ 
\\   $(v,s)$ & with $L(v) = \lor$ 
   & $\eloi$ & $E[v] \times \{ s \}$ 
\\   $(v,s)$ & with $L(v) = \land$ 
   & $\abel$ & $E[v] \times \{ s \}$ 
\\   $(v,s)$ & with $L(v) = \dia$ 
   & $\eloi$ & $E[v] \times R[s]$ 
\\   $(v,s)$ & with $L(v) = \Box$ 
   & $\abel$ & $E[v] \times R[s]$ 
\\ \hline
\end{tabular}
\end{center}
\caption{The evaluation game $\EG(\bbG,\bbS)$.}
\label{tb:2}
\end{table}

A central complexity measure for both parity formulas and modal automata will be the so-called {\em index}.
We define the index of a parity formula as the size of the range of its priority function $\Om$.  We will rely on the following result
from~\cite{kupk:size20} that ensures that throughout the paper we are able to 
work on parity formulas instead of formulas in standard syntax.

\begin{proposition}\label{prop:yeswecan}
    There is an algorithm that constructs for any formula $\varphi \in \muML$ an equivalent parity formula $\bbG_\varphi$ 
    such that $\size{\bbG_\phi} = \size{\varphi}$ and such that the index of $\bbG_\phi$ is smaller or equal to
    the alternation depth of $\varphi$. 
Conversely, there is an algorithm that constructs for a given parity formula 
$\bbG$ an equivalent formula $\varphi_\bbG \in \muML$ such that 
$\size{\varphi_\bbG} \leq 2 \cdot \size{\bbG}$ and such that the alternation 
depth of $\varphi_\bbG$ is smaller or equal to the index of $\bbG$.
\end{proposition}

\subsection{Modal Automata}

Intuitively, modal automata correspond to parity formulas in a certain normal 
form - the precise connection will be discussed in Section~\ref{sec:aut} below.
Modal automata are based on the \emph{modal one-step language}.
This language consists of very simple modal formulas, built up from 
a collection $A$ of \emph{variables}, which represent the states
of the automaton and correspond to the fixpoint  variables of a formula. 

\begin{definition}
\label{d:MLone}
Given a set $A$. 
The set $\MLone(A)$ of \emph{modal one-step formulas} over $A$ is inductively
given as follows:
\[
\al
\isbnf \bot \divbnf \top 
   \divbnf \dia a \divbnf \Box a
   \divbnf \al \land \al \divbnf \al \lor \al,
\]
where $a \in A$.
We let $\Sfor(\al)$ denote the collection of subformulas of a one-step formula
$\al$.
\end{definition}

\begin{definition}
\label{d:modaut}
Let $\Prop$ be a finite set of propositional variables. 
A \emph{modal $\Prop$-automaton} $\bbA$  is a quadruple $(A,\De,\Om,a_{I})$ where
$A$ is a non-empty finite set of \emph{states}, of which $a_{I} \in A$ is the 
\emph{initial} state, $\Om : A \to \omega$ is the \emph{priority map}, and the
\emph{transition map} 
$\De : A \times \pow\Prop \to \MLone(A)$
maps pair of states and colors to one-step formulas.
\end{definition}

The size of a modal automaton is defined as follows.

\begin{definition}
Let $\bbA = (A,\De,\Om,a_{I})$ be a modal automaton.
We define its \emph{state size} $\asz{\bbA} \isdef \size{A}$, 
its \emph{size} as 
\[
\size{\bbA} \isdef 
\left| \bigcup \{ \Sfor(\al) \mid \al \in \Ran(\De) \} \right| + \asz{\bbA},
\]
and its \emph{index} as $\idx(\bbA) \isdef \size{\Ran(\Om)}$.
\end{definition}

Modal automata operate on pointed models, acceptance is defined 
via parity graph games.

\begin{definition}
Let $\bbA = (A,\De,\Om,a_{I})$ be a modal automaton and let $(\bbS,s_{I})$ be 
a model. 
The acceptance game $\AG(\bbA,\bbS)$ of $\bbA$ has the game board displayed in
Table~\ref{fig:acceptance}. 
 \begin{table}
 \begin{center}
     \begin{tabular}{|c|c|c|}
     \hline 
      Position & Player & Admissible moves \\
      \hline \hline
      $(a,s) \in A \times S$ & $-$ & $\{(\De(a,c_s),s)\}  \mbox{ with } c_s = \{ p \in Q \mid s\in \V(p) \}$ \\
      \hline
      $(\bot,s) \in \MLone(A) \times S$ & $\exists$ & $\emptyset$ \\ 
      \hline 
      $(\top,s) \in \MLone(A) \times S$ & $\forall$ & $\emptyset$ \\
      \hline 
      $(\alpha_1 \lor \alpha_2,s)$ & $\exists$ & $\{ (\alpha_1,s),(\alpha_2,s) \}$ \\
      \hline
       $(\alpha_1 \land \alpha_2,s)$ & $\forall$ & $\{ (\alpha_1,s),(\alpha_2,s) \}$ \\
       \hline
      $(\dia a,s)$ & $\exists$ & $\{ (a,s') \mid s' \in R[s] \}$ \\
      \hline
      $(\Box a,s)$ & $\forall$ & $\{ (a,s') \mid s' \in R[s] \}$ \\
      \hline
     \end{tabular}
 \end{center}
 \caption{The game board of the accepance game of a modal automaton.}
 \label{fig:acceptance}
 \end{table}
A pointed model $(\bbS,s_I)$ is accepted by $\bbA$ if $\exists$ has a winning
strategy at position $(a_I,s_I)$ in $\AG(\bbA,\bbS)$.
\end{definition}

\subsection{Disjunctive Formulas \& Automata}

In this section we introduce disjunctive formulas and their automata-theoretic 
pendant, so-called disjunctive automata.
Disjunctive formulas can be best characterised in a modal language that has one
``cover modality'' $\nb$ that takes a finite set of formulas as its argument. 
Given such a set $\Phi$, one may think of the formula $\nb\Phi$ as the 
abbreviation 
\[
\nb\Phi \equiv 
\bigwedge_{\phi \in \Phi} \dia \phi \land 
\Box \bigvee_{\phi \in \Phi} \phi.
\]
It is called the ``cover modality'' since, intuitively, the formula $\nb\Phi$
holds at a point $s$ if the set of successors of $s$ and the set of elements of 
$\Phi$ cover each other, in a sense that can be made precise using the notion of 
\emph{relation lifting}.
For a relation $Z \subseteq X \times Y$ we define its lifting 
$\ol{Z} \subseteq \pow X \times \pow Y$ by putting 
\[
(U,V) \in \ol{Z} \quad \mbox{ if } \quad \forall x \in U. \exists y \in V. (x,y) \in Z \mbox{ and } 
 \forall y \in V. \exists x \in U. (x,y) \in Z.
\]
It is then easy to verify that the formula $\nb\Phi$ holds at a point $s$ if the
pair $(R[s],\Phi)$ belongs to the lifting of the truth relation between points
and formulas.
The operator $\nb$ is well-known from the literature, 
cf.~e.g.~\cite{jani:auto95,vene:auto06,leht:disj15}. 

\begin{definition}
 \label{d:Dparityformula}
Let $\Prop$ be a finite set of propositional variables. 
A {\em disjunctive parity formula} over $\Prop$ is a quintuple 
$\bbG = (V,E,L,\Om,v_{I})$ such that
\begin{itemize}
\item  $L: V \to \{\land_l \mid l \in \Lit(\Prop) \}
      \cup \{ \nb, \lor, \top, \epsilon \}$ is a labelling function;
\item for all $v \in V$ with $L(v) \not= \nb$ we have $\size{E[v]} \leq 2$
\end{itemize}
and such that all other conditions of the definition of parity formulas in
Def.~\ref{d:pf} are met.
The board of the evaluation game $\EG(\bbG, \bbS)$ of a disjunctive parity formula on a model $\bbS$ is displayed in Table~\ref{tb:3}. 
\begin{table}[t]
\begin{center}
\begin{tabular}{|ll|c|c|}
\hline
\multicolumn{2}{|l|}{Position} & Player  & Admissible moves 
\\\hline
    $(v,s)$ & with $L(v) = \epsilon$ 
   & - & $E[v] \times \{ s \}$ 
\\  $(v,s)$ & with $L(v) = \top$ 
   & $\abel$ & $\nada$ 
\\   $(v,s)$ & with $L(v) = \lor$ 
   & $\eloi$ & $E[v] \times \{ s \}$ 
\\   $(v,s)$ & with $L(v) = \land_l$ and $s \not\in \V(l)$
& $\eloi$ & $\nada$ 
\\   $(v,s)$ & with $L(v) = \land_{l}$ and $s \in \V(l)$
& $\abel$ & $E[v] \times \{ s \}$ 
\\   $(v,s)$ & with $L(v) = \nb$ 
   & $\eloi$ & $\{ Z \subseteq V \times S \mid (E[v],R[v]) \in \ol{Z} \}$ \\
      Z & $\subseteq V \times S$ & $\forall$ 
	& $\{ (v',s') \mid (v',s') \in Z \}$ \\
  \hline
\end{tabular}
\end{center}
\caption{The evaluation game $\EG(\bbG,\bbS)$.}
\label{tb:3}
\end{table}
\end{definition}
Intuitively, a node $v$ labelled with $\nb$ represents the formula 
$\nb \{ \phi_{w} \mid w \in E[v] \}$,
where for each $w \in E[v]$ we write $\varphi_w$ for the formula represented by
$w$. 
Furthermore $\land_l$ is intended to be a unary operator that represents the 
conjunction of its argument with $l$. 
In other words, disjunctive formulas are formulas that contain conjunctions only 
in the form of conjunctions with literals and in the form of $\top$ that can be
thought of as the empty conjunction. 
Disjunctive formulas have their automata-theoretic pendant, the so-call disjunctive
modal automata -- the so-called \emph{$\mu$-automata} of Janin \& 
Walukiewicz~\cite{jani:auto95}.
These are defined by restricting the shape of transition conditions. 


\begin{definition}
\label{d:DMLone}
Given a finite set $A$, we define the set $\DMLone(A)$ of \emph{disjunctive modal 
one-step formulas} over $A$ via the following grammar: 
\[
\al
\isbnf \bot \divbnf \top 
   \divbnf \nb B
   \divbnf \al \lor \al,
\]
where $B \sse A$. A disjunctive modal $\Prop$-automaton is a tuple
$\bbA = (A,\De,\Om,a_{I})$ such that
$\De: A \times \pow \Prop \to \DMLone(A)$.
The acceptance game $\AG(\bbA,\bbS)$ on a model $\bbS$ is defined
as for general modal automata with the difference that the rule for $\land$ no longer applies and
that the rules for $\Box$ and $\dia$ are replaced  by
\begin{center}
 \begin{tabular}{|c|c|c|}
     \hline 
      Position & Player & Admissible moves \\
     \hline \hline 
     $(\nb B,s)$ & $\exists$ & $\{ Z \subseteq A \times S \mid (B,R[s]) \in \ol{Z} \}$ \\
     \hline
     $Z \subseteq A \times S$ & $\forall$ & $\{ (v',s') \mid (v',s') \in Z \}$ \\
     \hline 
 \end{tabular}
\end{center}
\end{definition}

\section{Guarding Revisited}

Existing approaches for turning a $\mu$-calculus formula into a modal
automaton rely on the assumption that the input formula is guarded.
As \cite{brus:guar15} have shown this assumption is problematic because
existing algorithms for guarding formulas, which have long been thought
to be polynomial, are in fact exponential. In this section we discuss
two results on the complexity of guarding formulas. We do this in the
setting of parity formulas.

\begin{definition}
\label{d:guar}
A path $\pi = v_{0}v_{1}\cdots v_{n}$ through $\bbG$ is \emph{unguarded} if 
$n\geq 1$, $v_{0}, v_{n} \in \Dom(\Om)$ while there is no $i$, with $0 < i \leq 
n$, such that $v_{i}$ is a modal node. A
parity formula is \emph{guarded} if it has no unguarded cycles, and
\emph{strongly guarded} if it has no unguarded paths.
\end{definition}

Adapting the well-known construction for guarding formulas one can show
that it is possible to guard parity formulas, with an exponential
blow-up in the number of states~\cite{kupk:size20}.

\begin{theorem}
\label{t:guard1}
There is an algorithm that transforms a parity formula $\bbG = (V,E,L,\Om,v_{I})$
into a strongly guarded parity formula $\bbG^{g}$ such that 

\begin{enumerate}[topsep=0pt,itemsep=-1ex,partopsep=1ex,parsep=1ex,%
    label={\arabic*)}]

\item \label{eq:tg1:1} 
$\bbG^{g} \equiv \bbG$;

\item \label{eq:tg1:2}
$\size{\bbG^{g}} \leq 2^{1+\size{\Dom(\Om)}} \cdot \size{\bbG}$;

\item \label{eq:tg1:3}
$\idx(\bbG^{g}) \leq \idx(\bbG)$.
\end{enumerate}
\end{theorem}

It is unclear whether this result can be improved such that the number
of states of $\bbG^g$ is polynomial in the number of states in $\bbG$.
The results in section~4 of Bruse, Friedmann \& Lange~\cite{brus:guar15} show that 
certain guarded transformation procedures are as hard\footnote{%
   It is an open question whether parity games can be solved in polynomial time. 
   Despite considerable efforts no polynomial algorithm has been found so far.
   In the recent literature, however, various quasi-polynomial algorithms have 
   been given, following the breakthrough work of Calude et 
   alii~\cite{calu:deci17}.
   }
as solving parity games. Theorem~\ref{t:gtlow} below can be seen as our
parity-formula version of this observation. Our proof is in fact simpler
because we can exploit the close connection between parity games and
parity formulas and thus do not need the product construction from
\cite{kupfer:linbran05} that is used for the results from
\cite{brus:guar15}.

\newcommand{\ptime}{\textsc{ptime}}

\begin{theorem}
\label{t:gtlow}
If there is a procedure that runs in polynomial time and transforms a parity
formula $\bbG$ to a guarded parity formula $\bbG^\gamma$ with $\bbG^\gamma 
\equiv \bbG$ then solving parity games is in \ptime.
\end{theorem}

For a proof of Theorem~\ref{t:gtlow} we refer to~\cite{kupk:size20}, where we also
discuss in some detail the relation with other results in~\cite{brus:guar15}.
%

\section{Modal Automata and Strongly Guarded Parity Formulas}
\label{sec:aut}

In this section we establish a close connection between modal automata and 
parity formulas.
We will first see that a modal automaton can be turned into a strongly guarded 
parity formula that is of linear size if we ignore propositional variables. 
In particular this shows that turning a parity formula into an equivalent modal
automaton is at least as hard as the guarding construction from the previous 
section (hardness of the latter does not depend on formulas containing 
propositional variables).
We close by showing how to turn a parity formula
into an equivalent modal automaton of exponential size. 
Collectively the results from this section will show that 
constructing a disjunctive modal automaton from a parity formula
by first turning the latter into an equivalent modal automaton to which we then
apply a known ``non-determinisation'' construction would yield a doubly
exponential blow-up (in closure size). This sets the stage for our main result 
in the next section where a new construction that is single exponential in 
closure size from parity automata to disjunctive modal automata is provided.

\begin{theorem}\label{thm:modal}
There is an algorithm that constructs, given a modal automaton $\bbA$, a
strongly guarded parity formula $\bbG$ such that 

1) $\bbG \equiv \bbA$;

2) $\size{\bbG} \leq 2^{\size{\Prop}} \cdot \size{\bbA}$;

3) $\idx(\bbG) \leq \idx(\bbA)$.

4) If $\bbA$ is disjunctive then so is $\bbG$.
\end{theorem}

\begin{proof}
   We only sketch the construction.  
   For each $a \in A$ we let
   $\De'(a) = \bigvee_{c \in \pow \Prop} (\bigwedge_{p \in c} p \land \bigwedge_{p \not \in c} \lneg{p}  \land \De(a,c))$.
   Given a modal automaton $\bbA=(A,\De,\Om,a_I)$
   we define the set of nodes of $\bbG$ by
   \[ V = A  \cup \bigcup \{ \Sfor(\al) \mid \al \in \Ran(\De') \} \]
To defined the edge relation $E$ of $\bbG$ we put $E[a] = \De'(a)$ for all $a \in A$.
For all other elements of $V$ we let $E$ be the ``immediate subformula'' relation, e.g.
$E[\alpha_1 \land \alpha_2] = \{\alpha_1,\alpha_2\}$, $E[\dia a] = \{a\}$, etc. 
The priority map $\Om_\bbG$ of $\bbG$ assigns to each element $a \in A$ the priority $\Om(a)$. 
The initial state $v_I$ of $\bbG$ is defined as $v_I = a_I$.
Finally, the map $L$ assigns to each element $a \in A$ the label $\epsilon$ and for each formula
$\al$ the label consists of the top-most operator of $\al$. It is easy to see that $\bbG$ thus defined satisfies conditions 2) and 3). The proof of condition
   1) is following a standard argument and is omitted here.
   Concerning 4) it suffices to note that for a disjunctive $a \in A$ we put
   $$\De'(a) = \bigvee_{c \in \pow \Prop} \left(\land_{l_1} \left( \cdots \land_{l_m}(\De(a,c))\right)\right)$$
   where $l_1,\dots,l_n$ is an enumeration of all propositional 
   variables in $c$ and the negation of all variables 
   not in $c$. Otherwise $\bbG$ is defined as in the non-disjunctive case and the result is a disjunctive formula.
\end{proof}

\begin{theorem}
There is an algorithm that constructs, given a parity formula $\bbG$, a modal
automaton $\bbA$ such that 

1) $\bbA \equiv \bbG$;

2) $\asz{\bbA} \leq  2^{1+\size{\Dom(\Om)}} \cdot \size{\bbG}$ and 
   $\size{\bbA} \leq 2^{\size{\Prop} + 1+\size{\Dom(\Om)}} \cdot \size{\bbG}$;

3) $\idx(\bbA) \leq \idx(\bbG)$.
\end{theorem}
\begin{proof}
By Theorem~\ref{t:guard1} we may effectively construct from $\bbG$ an equivalent,
strongly guarded parity formula $\bbH = (V,E,L,\Om,v_{I})$ such that 
$\size{\bbH} \leq  2^{1+\size{\Dom(\Om)}} \cdot \size{\bbG}$ and 
$\idx(\bbH) \leq \idx(\bbG)$.
As shown in ~\cite{kupk:size20}, we may additionally assume that in $\bbH$,
every predecessor of a node $v \in \Dom(\Om)$ is a modal node.
The state space of the modal automaton $\bbA$ will be given as the set
$A \isdef E[V_{m}] \cup \{ v_{I} \}$, 
so that by the assumption on $\bbH$ every state of $\bbH$ is a state of $\bbA$.
In addition, $v_{I}$ and possible other successors of modal nodes are states of 
$\bbA$ as well.
We can then simply define $\Om_{\bbA} \isdef \Om$,
and take $v_{I}$ as the initial state of $\bbA$. 
It remains to define the transition function $\De$ of $\bbA$.


Our first step will be to associate, with each node $v \in V$, a formula
$\al(v)$, which belongs to the collection $\AMLone(\Prop,A)$ of 
\emph{alternative} one-step formulas given by the following grammar:
\[
\al
\isbnf \bot \divbnf \top 
   \divbnf p \divbnf \lneg{p}
   \divbnf \dia a \divbnf \Box a
   \divbnf \al \land \al \divbnf \al \lor \al,
\]
where $p \in \Prop$ and $a \in A$.
As the formula $\bbH$ is strongly guarded and by the 
additional property that $E^{-1}[\Dom(\Om)] \sse V_{m}$, there is a unique
map $\al: V \to \AMLone(\Prop,A)$ which satisfies the following conditions:
\[
\al(v) = \left\{ \begin{array}{lll}
   L(v) & 
    \text{if } v \in V_{l} 
\\ \hs a
   \quad  \text{where } L(v) = \hs \text{ and } E[v] = \{ a \} 
   &  \text{if } v \in V_{m} 
\\ \bigodot \big\{ \al(u) \mid u \in E[v] \big\}
    \; \text{where } L(v) = \odot
   & \text{if } v \in V_{b} 
\\ \al(u)
   \quad \text{where } E[v] = \{ u \}
   & \text{if } v \in V_{\epsilon} 
\end{array}\right.
\]

We can now define the transition map $\De: A \times \pow(\Prop) \to \MLone(A)$ as
follows.
For each state $a \in A$ and color $c \in \pow(\Prop)$ we define the formula
$\De(a,c)$ as $\De(a,c) \isdef \al(a)[\si_{c}]$,
where the substitution $\si_{c}: \AMLone(\Prop,A) \to \MLone(A)$ is given by 
putting
\[
\si_{c}(p) \isdef \left\{\begin{array}{ll}
   \top & \text{if } p \in c
\\ \bot & \text{if } p \not\in c.
\end{array}\right.
\]
This completes the definition of the automaton $\bbA$. 
It is easy to see that $\asz{\bbA} \leq \size{V} \leq 2^{1+\size{\Dom(\Om)}} \cdot \size{\bbG}$, that 
$\size{\bbA} \leq 2^{\size{\Prop}} \times V \leq  2^{\size{\Prop} +1 + \size{\Dom(\Om)}}   \cdot \size{\bbG}$
and that $\idx(\bbA) = \idx(\bbH) \leq \idx(\bbG)$,
which proves the items 2) and 3) of the theorem.
The equivalence of $\bbA$ and $\bbH$ (and thus, of $\bbA$ and $\bbG$) can be
proved by a routine argument.
\end{proof}

\section{The Simulation Theorem}

The main result of this section, and the main technical contribution of the paper, 
is the following theorem.

\begin{theorem}
\label{t:sim}
Let 
$\bbG$ be a parity formula of size $n$ and index $k$ with propositional
variables contained in $\Prop$ with $\size{\Prop} = l$.
Then we can effectively construct a disjunctive modal 
automaton\footnote{For the time being this will be an automaton with a regular
   acceptance condition. 
   We will transform this into an automaton with a parity condition later.
   }
$\bbA = (A,\Th,\Om,a_I)$ with
$\size{A} \leq 2^{n^{2}k}$, 
$\size{\bbA} \leq n 2^{n^{2}k + l + n}$ and 
$\bbG \equiv \bbA$. 
\end{theorem}

{\noindent \em Convention.} 
In the remainder of this section we fix a parity formula $\bbG$  
with $\bbG = (V,E,L,\Om,v_I)$ with $\size{V} = n$ and $\size{\Ran(\Om)}
= k$.
It will be convenient to make the following assumptions on $\bbG$:
(i) $\Om$ is total, 
(ii) $L^{-1}(\epsilon) = \nada$, and 
(iii) $E[v] \neq \nada$ if $L(v) \in \{ \land, \lor \}$.
We leave it for the reader to convince themselves that this is without loss of generality.
Furthermore we define $V_{\lor} \isdef L^{-1}(\lor)$, etc.
We now turn to the proof of Theorem~\ref{t:sim}. 

\subsection{Macrostates}

We shall construct the simulating automaton via a powerset construction.
That is, for the states of $\bbA$ we will in principle take subsets of $\bbG$.
However, in order to handle infinite matches correctly we need to store some
more information in these states: A state of $\bbA$ will be a a \emph{macrostate} over $\bbG$, 
that is, a ternary relation $m \sse V \times \Ran(\Om) \times V$, representing
various pieces of information.
Basically, each triple $(u,k,v) \in m$ represents the projection on $\bbG$ of
a partial play of the evaluation game of $\bbG$ which starts at $u$, ends at
$v$, and reaches $k$ as its highest priority (after $u$).
More precisely, the triple $(u,k,v)$ represents a path $\pi$ through $\bbG$ with 
$\first(\pi) = u$, $\last(\pi) = v$, and such that $k$ is the highest priority 
reached on $\pi$ after $u$.
Consequently, one single match of the acceptance game of $\bbA$ on a pointed 
Kripke structure will represent a certain bundle of matches of evaluation game 
of $\bbG$.

Before defining and discussing macrostates formally, we need some auxiliary
terminology.

\begin{definition}
\label{d:incst}
A subset $U \sse V$ is \emph{inconsistent} if there is 
$u \in U$ with $L(u) = \bot$, or if there are nodes $u,v \in U$ 
with $L(u) = p$ and $L(v) = \ol{p}$ for some $p \in \Prop$.
Given a color $c \in \pow\Prop$ we say that $U$ is \emph{compatible}
with $c$ if $L(u) \neq \bot$,
$L(u) = p$ implies $p \in c$, and $L(u) = \ol{p}$ implies $p \not\in c$, 
for all $u \in U$ and $p \in \Prop$.
\end{definition}

\begin{definition}
We define the set $M_\Om$ of \emph{macrostates} of $\bbG$  
by putting
$
M_{\Om} \isdef \pow(V \times \Ran(\Om) \times V)$.
The \emph{range} $\Ran(m)$ of a macrostate $m \in M_{\Om}$ is the set of all
$v \in V$ such that $(u,k,v) \in m$ for some $u \in V$ and $k \in \Ran(\Om)$.
With $m,m' \in M_{\Om}$, we define the \emph{composition} $m\comp m' \in 
M_{\Om}$ as the set of triples $(v,k,v'') \in V \times \Ran(\Om) \times V$ for
which we can find
$(v,k',v') \in m$ and $(v',k'',v'') \in m'$ such that $k = \max(k',k'')$.
For a subset $U \sse V$, we define $\Diag_{U} \isdef \{ (u,0,u) \mid u \in U \}$;
where $v \in V$, we abbreviate $\Diag_{v} \isdef \Diag_{\{ v \}}$.

A macrostate $m$ is called \emph{consistent}, respectively \emph{compatible 
with a color $c \in \pow\Prop$}, if $\Ran(m) \sse V$ satisfies the mentioned 
property w.r.t.~$c$, in the sense of Definition~\ref{d:incst}.

Given a stream $\al = (m_{i})_{i\in\om}$ of macrostates, we say that a stream 
$(v_{i},k_{i})_{i\in\om} \in (V\times\Ran(\Om))^{\om}$ is a \emph{trace on} $\al$
if $(v_{i},k_{i},v_{i+1}) \in m_{i}$, for all $i \in \om$.
Such a trace is \emph{good} (\emph{bad}, respectively) if the maximum number $k$ 
occurring as $k_{i}$ for infinitely many $i$ is even (odd, 
respectively).
We let $\NBT_{\Om}$ denote the collection of $M_{\Om}$-streams that do not 
carry a bad trace.
\end{definition}

\begin{proposition}\label{prop:parity}
The set $\NBT_{\Om}$ is an $\om$-regular language over $M_{\Om}$.
Concretely, there is a deterministic parity automaton $\bbP$ over
$M_{\Om}$ such that $\Lom(\bbP) = \NBT_{\Om}$ and $\bbP$ has size 
$2^{\Ord(nk\cdot \log(nk))}$ and index $\Ord(nk)$.
\end{proposition}
\begin{proof}
    We first observe that there is a non-deterministic parity word automaton
    $\bbW=(Q,q_I,\Delta_\bbW, \Om_\bbW)$ that accepts the language $\mathcal{L}_\mathrm{bad} \coloneqq 
    (M_{\Om})^\omega \setminus  \NBT_{\Om}$, i.e., all infinite streams
    of macrostates that do contain a bad trace. 
    To define $\bbW$ we put $Q \coloneqq V$, $q_I \coloneqq v_I$, $\Delta_\bbW(v,m) \coloneqq \{u \in V \mid \exists k. (v,k,u) \in m \}$
    and $\Om_\bbW(v) \coloneqq \Om(v) + 1$ for all $v \in V$ and all $m \in M_\Om$. It is easy to verify $\bbW$ is a parity automaton
    with $n$ states and index $k$ that accepts $\mathcal{L}_\mathrm{bad}$. Standard constructions can be used to first
    transform $\bbW$ into an equivalent non-deterministic B\"uchi automaton $\bbW'$ of size $\Ord(n \cdot k)$ (cf. e.g.~\cite{grae:auto02})
    which can be in turn transformed into an equivalent deterministic parity word automaton $\bbW''$
    that meets the size bounds of the proposition (cf.~\cite{pite:from07,sche:tigh09}). The automaton $\bbP$ is now constructed as the deterministic 
    parity word automaton
    that accepts the complement of the language of $\bbW''$ by adding 1 to all the priorities of states in
    $\bbW''$.
\end{proof}

\subsection{Local strategies \& the disjunctive modal automaton}

In our approach of dealing with the possible unguardedness of the input parity
formula, the key concept is that of a (positional) \emph{local strategy} for 
$\eloi$.
A local strategy represents a complete set of choices of $\eloi$ for all 
disjunction nodes in $\bbG$.
Intuitively, one may think of a local strategy as some part of a positional 
strategy of $\eloi$ where we stay put at a point in the model.
More precisely, a local strategy $\chi$ induces, in the evaluation game of 
$\bbG$, for any point in the model and any vertex in $\bbG$, a unique (partial)
play that does not leave the mentioned point and stops whenever a modal vertex 
in $\bbG$ is met.
The projections of these matches
will be called \emph{stationary plays}.
Formally we define these notion, together with some related concepts that we 
will discuss in a moment, as follows.

\begin{definition}
\label{d:ls}
A \emph{local strategy} on $\bbG$ is a map $\chi: V_{\lor} \to V$ such that 
$\chi(v) \in E[v]$, for all $ v \in V_{\lor}$. The collection of local 
strategies on $\bbG$ is denoted by $\LS_{\bbG}$.

Now fix such a local strategy $\chi$.
Given a vertex $v \in V$ we define the set $\mathit{SP}_{\chi}(v)$ of 
\emph{stationary $\chi$-plays from $v$} as the smallest set $S \sse V^{*}$
such that 
\\(1a) if $v \in V_{\lor}$ then $v\cdot\chi(v) \in S$; 
\\(1b) if $v \in V_{\land}$ then $v\cdot w \in S$, for each $w \in E[v]$;
\\(2a) if $\pi \in S$ and $u \isdef \last(\pi) \in V_{\lor}$, then 
   $\pi\cdot\chi(u) \in S$; and 
\\(2b) if $\pi \in S$ and $u \isdef \last(\pi) \in V_{\land}$, then
   $\pi\cdot w \in S$, for each $w \in E[u]$.
\\Given $\pi = v v_{1}\cdots v_{k} \in \mathit{SP}_{\chi}(v)$, define 
$\wt{\Om}(\pi) \isdef \max\{ \Om(v_{i}) \mid 1 \leq i \leq k \}$.
Via these stationary plays we define the following macrostates:   
\[\begin{array}{llll}
e_{\chi}^{-} & \isdef &
 \big\{ (v,n,u)  & \mid v \in V_{b} \text{ and } u = \last(\pi)
   \text{ for some } \pi \in \textit{SP}_{\chi}(v) 
  \text{ with } n = \wt{\Om}(\pi)
  \big\} 
\\ e_{\chi} & \isdef & e_{\chi}^{-} \cup \Diag_{V}
\end{array}\]
We say that on a macrostate $m$, $\chi$ is \emph{locally compatible} with a
color $c \in \pow\Prop$ if (i) $\Ran(m \comp e_{\chi})$ is compatible with $c$
and (ii)  the stream $m\comp (e^{-}_{\chi})^{\om}$ does not contain any bad trace.
\end{definition}

Here are some intuitions about these notions.
First, note that $\mathit{SP}_{\chi}(v)$ only contains paths of length $> 1$, 
corresponding to matches where actually a move is made at $v$.
The macrostate $e_{\chi}^{-}$ represents the relevant information about all 
finite stationary $\chi$-plays; hence, the collection of all \emph{infinite}
stationary $\chi$-plays corresponds to the set of all traces over the 
$M_{\Om}$-stream $(e_{\chi}^{-})^{\om}$.

Now fix a macro-state $m$. 
The macrostate $m \comp e_{\chi}^{-}$ represents the version of $m$ that absorbs
all continuations of matches in $m$ with one of these finite stationary 
$\chi$-plays; thus the set $m \cup (m \comp e_{\chi}^{-})$ represents all triples
in $m$ that are possibly continued with such a play.
For technical reasons it is convenient to define this set using the macro-state 
$e_{\chi}$, in the sense that $m \comp e_{\chi} 
= m \comp (e_{\chi}^{-} \cup \Diag_{V}) 
= (m \comp e_{\chi}^{-}) \cup (m \comp \Diag_{V}) = (m \comp e_{\chi}^{-}) \cup 
m$.
The stream $m \comp (e_{\chi}^{-})^{\om}$ represents all \emph{infinite} plays
that start with an $m$-play, and then continue with an infinite stationary 
$\chi$-play.
To say that there is a bad trace on such a stream means that for some $v \in 
\Ran(m)$, $\eloi$'s \emph{opponent} $\abel$ has a strategy such that, played 
against $\chi$, the resulting match (path through the graph of $\bbG$) is bad, 
in the sense of the highest priority met infinitely often being odd.
To say that, relative to $m$, $\chi$ is locally compatible with a color $c$ 
indicates, roughly, that after any $m$-match, if the local point of the Kripke 
model has color $c$, then it is safe for $\eloi$ (until the next modal vertex
in $\bbG$ is encountered) to continue by playing $\chi$.
Finally, it may also be useful to observe that for all $m$ and $\chi$ we find
$\Ran(m) \cap V_{p} \sse \Ran(m\comp e_{\chi}) \cap V_{p}$ and
$\Ran(m) \cap V_{m} \sse \Ran(m\comp e_{\chi}) \cap V_{m}$.
\medskip

We turn to the definition of the disjunctive modal automaton $\bbA_{\bbG}$.
For the definition of its transition map $\Th$, note a macrostate $m$ may be
thought of as representing the conjunction of states in $\Ran(m)$ that are 
visited by ``parallel'' plays of the evaluation game for $\bbG$. 
The transition function $\Th$ 
will implement the intuition that
a modal step needs to satisfy the ``demands'' posed by the modal nodes in 
$\Ran(m)$.
These demands are formulated separately for all box nodes and for each individual
diamond node.

\begin{definition}
Let $m \in M_{\Om}$ be some macrostate, and let $x \in \Ran(m) \cap 
V_{\dia}$.
Then we define 
\[\begin{array}{lll}
d_{\Box}(m)  &\isdef& 
   \{ (u,\Om(v),v) \mid u \in \Ran(m) \cap V_{\Box} \text{ and }
   v \in E[u] \},
\\ d_{x}(m)  &\isdef& 
   \{ (x,\Om(v),v) \mid v \in E[x] \} \cup d_{\Box}(m).
\end{array}\]
\end{definition}

The macrostates $d_{\Box}(m)$ and $d_{x}(m)$ correspond to, respectively, the 
\emph{universal} and \emph{existential} requirements made by the vertices in 
the range of $m$.

\begin{definition}\label{def:simulauto}
Let $\bbG = (V,E,L,\Om,v_I)$ be a parity formula.
We define the automaton $\bbA_{\bbG} = (A,\Th, \Acc, m_I)$ as follows.
To start with, its carrier is the collection of all macrostates:
$A \isdef M_{\Om}$
and its initial state is given as $m_I = \Diag_{v_{I}} = \{ (v_I,0,v_I)\}$.
The transition map $\Th$ is given as follows.
For a macrostate $m \in M_{\Om}$ and a local strategy $\chi$ we  define
 $A_{m,\chi} \sse M_{\Om}$:
\begin{eqnarray*}
A_{m,\chi} & \isdef & \{ e_{\chi} \comp d_{\Box}(m \comp e_{\chi}) \} 
\;\cup\;
   \{ e_{\chi} \comp d_{x}(m \comp e_{\chi}) \mid 
      x \in \Ran(m\comp e_{\chi}) \cap V_{\dia}\},
\end{eqnarray*}
and for a macrostate $m \in M_{\Om}$ and a color $c \in \pow\Prop$ we put: 
\[ \Th(m,c)  \isdef \bv \big\{ \theta(m,.c,\chi)
   \mid \chi \in \LS_{\bbG} \text{ is locally compatible with $c$ on } m 
   \big\},
\]
where
\[ \theta(m,c,\chi) \isdef 
\left\{ \begin{array}{ll}
   \nb A_{m,\chi}
   & \text{if } \Ran(m\comp e_{\chi}) \cap V_{\dia} \neq \nada
\\ \nb A_{m,\chi} \lor \nb\nada 
   & \text{if } \Ran(m\comp e_{\chi}) \cap V_{\dia} = \nada.
\end{array}\right.
\]
Finally, for its acceptance condition $\Acc$ we take the $\om$-regular language
$\NBT_{\Om}$, i.e, an infinite play of $\AG$ will be winning for $\eloi$ if
the corresponding stream of macrostates is in $\NBT_\Om$.
\end{definition}

To get some intuitions: to define $\Th(m,c)$, we nondeterministically \emph{guess}
a local strategy $\chi$ that is compatible with $c$ on $m$ -- this 
guess is reflected by the disjunction in the definition of $\Th(m,c)$.
For each such $\chi$, we absorb its stationary plays into $m$ and turn to the 
set of modal nodes in the range of the resulting macro-state $m \comp e_{\chi}$.
We gather the universal and existential requirements of $m \comp e_{\chi}$ into
an appropriate collection $A_{m,\chi}$ of ``next'' macro-states.
This set $A_{m,\chi}$ is then to be \emph{covered} by the collection of 
successors of the point in the Kripke model under inspection, as encoded by the 
formula $\nb A_{m,\chi}$.
In the special case where $m \comp e_{\chi}$ makes no existential demands (i.e., 
if $\Ran(m\comp e_{\chi}) \cap V_{\dia} = \nada$), we add the disjunct 
$\nb\nada \equiv \Box\bot$, allowing for the possibility that the point has no 
successors at all.
To see how this all works out precisely, the reader is advised to look at the
proof of Proposition~\ref{p:simeq} below.

\subsection{Proof of Theorem~\ref{t:sim}}
Turning to the proof of the main theorem, we first show that the disjunctive
modal automaton from Definition~\ref{def:simulauto} is equivalent to the parity
formula $\bbG$ that we started with. After that we prove Theorem~\ref{t:sim} by
showing that $\bbA_{\bbG}$ is within the desired size bounds.

\begin{proposition}
 \label{p:simeq}
For any parity formula $\bbG$, we have
$
\bbG \equiv \bbA,
$
where $\bbA \coloneqq \bbA_{\bbG}$ is given as in Definition~\ref{def:simulauto}. 
\end{proposition}

\begin{proof}
   We show that \vspace{-2mm}
\begin{equation}
\label{eq:simeq}
(v_{I},s_{I}) \in \Win_{\eloi}(\EG(\bbG,\bbS)) \text{ iff }
(m_{I},s_{I}) \in \Win_{\eloi}(\AG(\bbA,\bbS)),
\end{equation}
where $(\bbS,s_{I})$ is an arbitrary but fixed pointed model, $m_{I} \isdef 
\Diag_{v_{I}}$, and we write $\Win_{\eloi}$ for the set of winning positions
for $\eloi$ in a game.
In the sequel we will abbreviate $\EG \isdef \EG(\bbG,\bbS)$ and $\AG 
\isdef \AG(\bbA,\bbS)$.

For the direction from left to right of \eqref{eq:simeq}, fix a positional 
winning strategy $f$ for $\eloi$ in $\EG$.
For any point $s \in S$, we may associate a map $\chi_{s}: V_{\lor} \to V$ 
with $f$, as follows.
Given a vertex $u \in V_{\lor}$, if $(u,s) \in \Win_{\eloi}(\EG)$ then 
$\chi_{s}$ maps $u$ to the element $v \in E[u]$ such that $(v,s) = f(u,s)$; 
if $(u,s) \not\in \Win_{\eloi}(\EG)$ we define $\chi_{s}(u)$ to be an arbitrary 
element in $E[u]$. Clearly $\chi_{s}$ is a local strategy in the sense of
Definition~\ref{d:ls}, and it is not hard to prove that $\chi_{s}$ is locally 
compatible with the color $\V(s)$ of $s$, on any $m \in M_{\Om}$ such that 
$(v,s) \in \Win_{\eloi} (\EG)$ for any $v \in \Ran(m)$.
We define the following (positional) strategy $f'$ for $\eloi$ in $\AG$.
Let $\Si$ be a partial $\AG$-match with $\last(\Si) = (m,s)$.
In case $\Ran(m) \times \{s\} \not\sse \Win_{\eloi}(\EG)$ then $\eloi$ plays 
randomly (one may show that this will never happen).
If, on the other hand, $\Ran(m) \times \{s\} \sse \Win_{\eloi}(\EG)$, then we already saw
that $\chi_s$ is locally compatible with $\V(s)$ on $m$.
Compute $e \isdef e_{\chi_{s}}$ and recall that with each element $(u,k,v)
\in e$ we may associate a partial $f$-guided match $\pi: (u,s) \cdots (v,s)$,
which is stationary at $s$ and such that $k$ is the highest priority met on
$\pi$ after $u$: $k = \wt{\Om}(\pi)$.
It is then not hard to see that, for an arbitrary element $w \in \Ran(m \comp e)$
we have $(w,s) \in \Win_{\eloi}(\EG)$.
In particular, if such a $w$ belongs to $V_{\dia}$, then $\eloi$'s winning
strategy $f$ in $\EG$ selects a successor $t_{w} \in R[s]$.
This assignment, $V_{\dia} \ni w \mapsto t_{w}$, determines $\eloi$'s strategy
$f'$ in $\AG$.
That is, at any partial match $\Si$ with $\last(\Si) = (m,s)$, let 
$\eloi$ play 
\begin{equation}
\label{eq:xx1}
f'(m,s) \isdef 
\Big\{ \big(e \comp d_{\Box}(m\comp e), t\big) \mid t \in R[s] \Big\} \cup \Big\{ \big(e \comp d_{w}(m\comp e), t_{w}\big) \mid 
  w \in V_{\dia} \cap \Ran(m \comp e) \Big\}
.
\end{equation}
It is easy to see that this move is legitimate. 
To show that $f'$ is  actually winning for $\eloi$ 
consider an arbitrary $f'$-guided partial match $\Si = (m_{0},s_{0})\cdots(m_{n},s_{n})$ with $(m_{0},s_{0}) = 
(m_{I},s_{I})$.
Via an inductive proof one can show that 
$\Ran(m_{i}) \sse \{ v \in V \mid (v,s_{i}) \in \Win_{\eloi}(\EG)\}$,
for each $i \leq m$.
Therefore $\eloi$'s moves in $\Si$ are indeed legitimately guided as 
in~\eqref{eq:xx1} and she wins every finite $f'$-guided match of $\AG$ starting at $(m_{I},s_{I})$.

To see that $\eloi$ also wins the infinite $f'$-guided matches, let $\Si = 
(m_{i},s_{i})_{i\in\om}$ be an arbitrary such match, and consider an arbitrary 
trace $(v_{i},k_{i})_{i\in\om}$ on $\Si$.
Let $k$ be the maximum number occurring as $k_{i}$ for infinitely many $i$;
it then suffices to show that $k$ is even.
For every $i < \om$ one may associate
$f$-guided partial $\EG$-matches $\si_{i} =
(u^{i}_{0},s_{i})(u^{i}_{1},s_i) \cdots (u^{i}_{n_{i}},s_{i})$
and $\si_{i}^{+} \isdef \si_{i}\cdot(u^{i+1}_{0},s_{i+1})$
such that $v_{i} = u^{i}_{0}$, $v_{i+1} = u^{i+1}_{0}$ and $k = 
\wt{\Om}(\si_{i}^{+})$.
Putting these partial plays together, with the trace $(v_{i},k_{i})_{i\in\om}$
we have thus associated a (full)
infinite $f$-guided $\EG$-match $\si = \si_{0}\si_{1}\cdots$, such that 
$\si_{i} = (u^{i}_{0},s_{i})(u^{i}_{1},s_{i}) \cdots (u^{i}_{n_{i}},s_{i})$
and $k_{i} = \max\big\{ \Om(\si^{i}_{1}),\ldots,\Om(\si^{i}_{n_{i}}),
\Om(\si^{i+1}_{0})\big\}$.
This means that $k$ is the highest priority that occurs infinitely often in $\si$,
and since $\si$ is guided by $\eloi$'s winning strategy $f$, $k$ must be even
indeed.

For the other direction of \eqref{eq:simeq}, fix a winning strategy $h$ 
for $\eloi$ in $\AG$.
W.l.o.g. we may assume that $\bbS$ is an $\omega$-expanded
tree\footnote{%
   A pointed model $(\bbS,s_{I})$ is \emph{$\omega$-expanded} if $R$ is the 
   parent-child relation of a tree $(S,R)$ which has $s_{I}$ as its root, 
   and is such that every non-root node $s$ in $\bbS$ has at least $\omega$ 
   many bisimilar siblings.
   It is not hard to see that every pointed model can be unravelled to a 
   bisimilar model that is $\omega$-expanded.
   }
with root $s_{I}$, so that with each $s \in S$ we may associate a 
\emph{unique} state $m_{s}$ such that $(m_{s},s)$ can be reached during an 
$h$-guided match of $\AG$ starting from $(m_{I},s_{I})$.
By definition of $\bbA_{\bbG}$ and $\AG$ then, with each $s \in S$ we may also 
associate a local strategy $\chi_{s}$, which is locally compatible with 
the color $\V(s)$ on $m_{s}$ and such that $\eloi$'s strategy $h$ 
is aimed at satisfying the one-step formula $\nb A_{m_{s},\chi_{s}}$.

To define $\eloi$'s strategy $h'$ in $\EG$, consider an arbitrary finite 
$\EG$-match $\si$.
It is not hard to see that $\si$ admits a unique \emph{modal decomposition} $\si
= \si_{0}\cdots \si_{l}$, where for all $i < l$, $\last(\si_{i})$ is the 
unique modal position in $\si_{i}$, and $\si_{l}$ either contains no modal 
positions, or a unique one at $\last(\si_{l})$.
This means that we may present each $\si_{i}$ as $\si_{i} = 
(v^{i}_{0},s_{i})(v^{i}_{1},s_{i}) \cdots (v^{i}_{n_{i}},s_{i})$
for some fixed point $s_{i}$ in $\bbS$.
The key idea underlying the definition of $h'$ is that with every $h'$-guided 
finite match $\si$, with $\si = \si_{0}\cdots \si_{l}$ as above, we associate
an $h$-guided $\AG$-match $\Si_{\si} = (m_{0},s_{0})\cdots(m_{l},s_{l})$ 
satisfying the condition (\dag) given below:
\begin{enumerate}
\item[(\dag1)] for each $i \leq l$, $j \leq n_{i}$, we have 
$v^{i}_{j} \in \Ran(m_{i}\comp e_{\chi_{s_{i}}})$.
\item[(\dag2)] for each $i \leq l$, $m_{i} = m_{s_{i}}$, and the sequence
   $v^{i}_{0}\cdots v^{i}_{n_{i}}$ is $\chi_{s_{i}}$-guided;
   for each pair $j,k$ with $j < k \leq n_{i}$, we have 
   $(v^{i}_{j},N^{i}_{j,k},v^{i}_{k}) \in e^{-}_{\chi_{s_{i}}}$,
   where $N^{i}_{j,k} = \max\{ \Om(v^{i}_{r}) \mid j < r \leq k \}$;
\item[(\dag3)] for each $i < l$, with $w \isdef v^{i}_{n_{i}}$, 
   if $w \in V_{\dia}$, 
   then $m_{i+1} = e_{\chi_{s_{i}}} \comp d_{w}(m_{i}\comp e_{\chi_{s_{i}}})$
   and $(v^{i}_{0},M_{i},v^{i+1}_{0}) \in m_{i+1}$, 
   where $M_{i} = \max\{ \Om(v^{i}_{1}), \ldots, \Om(v^{i}_{n_{i}}), \Om(v^{i+1}_{0}) \}$.
\item[(\dag4)] for each $i < l$, with $w \isdef v^{i}_{n_{i}}$, 
   if $w \in V_{\Box}$, 
   then $m_{i+1} = e_{\chi_{s_{i}}} \comp d_{\Box}(m_{i}\comp e_{\chi_{s_{i}}})$
   and $(v^{i}_{0},M_{i},v^{i+1}_{0}) \in m_{i+1}$, 
   where $M_{i} = \max\{ \Om(v^{i}_{1}), \ldots, \Om(v^{i}_{n_{i}}), \Om(v^{i+1}_{0}) \}$.
\end{enumerate}

Based on this connection, we define the following strategy $h'$ for $\eloi$ in 
$\EG$; we show at the same time that, playing $h'$, $\eloi$ can maintain the 
condition (\dag) and wins all finite matches.
Consider a partial $h'$-guided match $\si$, modally decomposed as $\si = 
\si_{0}\cdots \si_{l}$ as above, where $\si_{l} = (v^{l}_{0},s_{l})\cdots 
(v^{l}_{k},s_{l})$, and let $\Si_{\si} = (m_{0},s_{0})\cdots(m_{l},s_{l})$
be an associated $\AG$-match satisfying (\dag).
We distinguish cases, writing $v \isdef v^{l}_{k}$ and $\chi \isdef 
\chi_{s_{l}}$ for brevity.
\begin{itemize}
\setlength{\itemsep}{0mm}
\setlength{\parsep}{0mm}
\item
If $v$ is a propositional node, we need to show that $\si$ is won by $\eloi$.
This is immediate in case $L(v) = \top$, so assume that $L(v) = \bot$ or
$L(v) \in \{ p, \ol{p} \}$ for some proposition letter $p$.
We only treat the case where $L(v) = p$, the other cases being similar.
By (\dag2) we have $m_{l} = m_{s_{l}}$, so that 
$\chi = \chi_{s_{l}}$ is locally compatible with the color $\V(s_{l})$
on $m_{l}\comp e_{\chi}$.
But then (\dag1) implies that $L(v) \in \V(s)$.

\item
If $v \in V_{\lor}$, define $h'(\si) \isdef (\chi(v),s_{l})$.
It is easy to see that $\si\cdot (\chi(v),s_{l})$ and $\Si_{\si}$ are 
related by (\dag).
\item
If $v \in V_{\land}$, suppose that $\abel$ picks a conjunct $u$ of $v$.
Then $\si\cdot (u,s_{l})$ and $\Si_{\si}$ are related
by (\dag).

\item
If $v \in V_{\dia}$, first define $m_{l+1} \isdef e_{\chi} \comp 
d_{v}(m_{l}\comp e_{\chi})$.
Note that, since $v \in \Ran(m_{l}\comp e_{\chi})$ by the inductive hypothesis
(\dag1), we find $m_{l+1} \in A_{m_{l},\chi}$.
Furthermore, recall that by our assumption on $h$, $\eloi$'s move at position
$(m_{l},s_{l})$ in $\AG$ is aimed at satisfying the one-step formula $\nb 
A_{m_{l},\chi}$, and so this move must contain a pair of the form 
$(m_{l+1},t)$ 
for some $t \in R[s_{l}]$.
Now define 
$h'(\si) \isdef (u,t)$, 
where $u$ is the (unique) element of $E[v]$.
The modal decomposition of $\si' \isdef \si \cdot (u,t)$ is then $\si' = 
\si_{1}\cdots\si_{m}\si_{m+1}$, where $\si_{m+1} = (u,t)$.
(That is, in the terminology of (\dag) we have $v^{l+1}_{0} = u$ and $s_{l+1}
= t$.)

We now check that $\si'$ and $\Si' \isdef \Si \cdot (m_{l+1},t)$ are related
by (\dag).
For (\dag1), it suffices to show that $u \in \Ran(m_{l+1}\comp e_{\chi_{t}})$.
But this is immediate by
$(v,\Om(u),u) \in 
d_{v}(m_{l}\comp e_{\chi}) =
\Diag_{V} \comp d_{v}(m_{l}\comp e_{\chi}) \sse 
e_{\chi} \comp d_{v}(m_{l}\comp e_{\chi}) =
m_{l+1} \sse
m_{l+1} \comp \Diag_{V} \sse
m_{l+1} \comp e_{\chi_{t}}$.
For (\dag2) it suffices to show that $m_{l+1} = m_{t}$ but this holds by
construction.
For (\dag3) we likewise have $m_{l+1} = e_{\chi}\comp d_{v}(m_{l}\comp e_{\chi})$
by construction.
We already saw that $(v,\Om(u),u) \in d_{v}(m_{l}\comp e_{\chi})$, and we have 
$(v^{l}_{0},N^{l}_{0,k},v^{l}_{k}) \in e_{\chi}^{-}$ by the induction hypothesis 
(\dag2).
From this, and the observation that $M_{l} = \max(N^{l}_{0,k},\Om(u))$ we obtain 
$(v^{l}_{0},M_{l},u) \in e_{\chi}^{-}\comp d_{v}(m_{l}\comp e_{\chi}) \sse
e_{\chi} \comp d_{v}(m_{l}\comp e_{\chi}) = m_{l+1}$.
Finally, for (\dag4) there is nothing to prove.

\item
The case where $v \in V_{\Box}$ is similar to the previous one, so we skip some
details.
Let $u \in E[v]$ be the unique successor of $v$ in $\bbG$, and suppose that
in our $\EG$-match, $\abel$ picks a successor $t$ of $s_{l}$; that is, we now 
look at the continuation $\si' \isdef \si \cdot (u,t)$ of the $\EG$-match 
$\si$.
Consider $\eloi$'s move in $\AG$ at position $(m_{l},s_{l})$, which makes the 
one-step formula $\nb A_{m_{l},\chi}$ true, and thus contains a pair $(m,t)$ 
for some $m \in A_{m_{l},\chi}$.
Now define $m_{l+1} \isdef m$, and let $\Si' \isdef (m,t)$; this is an
$h$-guided continuation of $\Si$.

To verify that $\si'$ and $\Si'$ satisfy (\dag), first note that by definition of 
the set $A_{m_{l},\chi}$, $m$ must be of the form $e_{\chi} \comp d_{x}(m_{l}
\comp e_{\chi})$, where either $x = \Box$ or $x \in \Ran(m_{l}\comp e_{\chi})$.
Based on this observation, checking the conditions (\dag1), (\dag2) and (\dag4)
are similar to the respective conditions (\dag1), (\dag2) and (\dag3) in the 
previous case (with the only difference that we now must also take the 
possibility that $x = \Box$ into account).
Finally, condition (\dag3) needs no check since it is not applicable.
\end{itemize}
To see that $h'$ is winning for $\eloi$
consider a full $h'$-guided match, and distinguish cases. For finite matches one can check
that $\exists$ never gets stuck. 
In case $\si$ is an infinite $h'$-guided match, we make a further 
distinction as to whether the number of modal positions that $\si$ passes 
through is finite or infinite. If $\si$ passes through infinitely many modal positions, there is a unique 
way of decomposing $\si$ as $\si = \si_{0}\si_{1}\cdots$, where 
$(\last(\si_{i}))_{i\in\om}$ is the sequence of (all) modal positions in $\si$.
By construction there is an associated \emph{infinite} $h$-guided 
$\AG$-match $\Si_{\si} = (m_{i},s_{i})_{i\in\om}$ related to $\si$ via the
condition (\dag).
It is now possible to prove
that $\eloi$ is the winner of $\sigma$ by using that $h$ is winning for $\eloi$ in $\AG$ (details omitted due
to space limitations).
%
If $\si$ only passes finitely many modal positions, 
we may represent $\si = \si_{0}\cdots\si_{l}$, where each $\si_{i}$ with $i<l$ 
is finite, $\si_{l}$ is infinite, and $(\last(\si_{i}))_{i<l}$ is the sequence 
of all modal positions in $\si$.
We only consider the subcase where $l > 0$. 
Let $\Si_{\si} = (m_{0},s_{0})\cdots(m_{l},s_{l})$ be the $h$-guided
$\AG$-match that we have associated with $\si$ (or, to be more precise, with
the initial segments of $\si$ that are long enough to have passed the last 
modal node of $\si$).
Observe that since $\Si_{\si}$ is $h$-guided, the position $(m_{l},s_{l})$ must 
be winning for $\eloi$, and that by~(\dag2) the macrostate $m_{l}$ 
is the unique state $m$ in $A$ such that $(m_{l},s_{l})$ is met during
an $h$-guided $\AG$-match.
Write $\si_{m} = 
(u_{0},s)(u_{1},s)(u_{2},s)\cdots$; that is, we write $u_{j} \isdef v^{l}_{j}$
and $s \isdef s_{l}$.
The sequence $u_{0}u_{1}u_{2}\cdots$ is a 
trace on the stream $m_{l}\comp (e^{-}_{\chi_{s}})^{\om}$; but then, by 
the compatibility of $\chi_{s}$ with $m_{l}$ on $\V(s)$,
$u_{0}u_{1}u_{2}\cdots$ must be a \emph{good} trace. 
Since $(u_{0},s)(u_{1},s)\cdots$ is a tail of $\si$, this means 
that $\si$ is won by $\eloi$, as required.
\end{proof}

\begin{proofof}{Theorem~\ref{t:sim}}
The equivalence part of the disjunctive modal automaton $\bbA_{\bbG,v_I}$ to $\bbG$ was 
proved in Proposition~\ref{p:simeq}. It remains to check the sizes of the components of the automaton $\bbA$. 
But this is fairly straightforward.
To start with, from the definition of $\bbA$ we have $M_{\Om}  = \pow(V \times \Ran(\Om)
\times V)$ it immediately follows that 
$
\size{A} \leq 2^{\size{V \times \Ran(\Om) \times V}} = 2^{n^{2}k}.
$
To compute the size of $\Th$, first observe that the number of local strategies 
is equal to $2^{\size{A_{\lor}}}$, and that for each macrostate $m$, 
local strategy $\chi$ and color $c \in \pow\Prop$ we find
$\size{A_{m,\chi}} \leq \size{V_{\dia}} +1 \leq n$.
From this it is immediate that for each formula $\Th(m,c)$ we have
$
\size{\Th(m,c)} \leq 2^{\size{A_{\lor}}} \cdot (\size{V_{\dia}} +1)
\leq n 2^{n}$.
Finally, the table of $\Th$ has $\size{A} \cdot 2^{\size{\Prop}} \leq 
2^{n^{2}k} \cdot 2^l = 2^{n^{2}k + l}$ entries, so that its total size is bounded by 
$n 2^{n} \cdot 2^{n^{2}k + l} = n 2^{n^{2}k + l + n}$, as stated by the theorem.
\end{proofof}

\begin{corollary}\label{cor:parity}
    The disjunctive modal automaton $\bbA_{\bbG,v_I}$ 
    can be turned into an equivalent disjunctive parity automaton $\bbA$ with index $\Ord(n\cdot k)$ and size 
    $2^{\Ord(n^2k\cdot \log(nk))}$.
\end{corollary}
\begin{proof}
A standard construction, the so-called wreath product, can be used to turn the automaton $\bbA_{\bbG,v_I}$ together
with the automaton $\bbP=(P,\delta,\Om_P,p_I)$ from Proposition~\ref{prop:parity} into a parity automaton (cf.~e.g.~\cite[Definition~4.3]{kuve08:coal}). 
The transition map of the resulting
automaton $\bbA$ will have the same size as the one of $\bbA_{\bbG}$, 
the set of states is given as the product
$M_\Om \times P$ and the index of $\bbA$ is equal to the index of $\bbP$. 
Hence the parity automaton has size
$2^{\Ord(nk\cdot \log(nk))} \cdot  n 2^{n^{2}k + l + n} 
= 2^{\Ord(n^2k\cdot \log(nk))}$,
and index $\Ord(nk)$.
%
\end{proof}

We finish with the main result of our paper: there is an algorithm turning a parity formula into an equivalent 
disjunctive one in exponential size of the original formula. Due to our size preserving translations from
parity formulas to $\mu$-calculus formulas in the standard syntax, the result carries directly over to formulas in standard syntax if we measure the size of this formula in terms of its closure. 

\begin{corollary}
    For any parity formula $\bbG$ we can construct an equivalent disjunctive parity formula $\bbG^d$ with 
     $\size{\bbG^d} \leq 2^{\Ord(n^2k\cdot \log(nk))}$ and with index $\Ord(n\cdot k)$. Here $n =\size{\bbG}$ and $k$ is the index 
     of $\bbG$. 
\end{corollary}

The corollary is an immediate consequence of Corollary~\ref{cor:parity} and Theorem~\ref{thm:modal}.
An application of Prop.~\ref{prop:yeswecan} shows that the corollary implies Thm.~\ref{thm:main}. 

\section{Conclusions}

We have presented an algorithm that constructs for a given arbitrary formula in the modal $\mu$-calculus 
an equivalent disjunctive formula with a single exponential blow-up when measuring the size of a formula 
in closure size. While the complexity of this construction is likely to be optimal, it is an interesting
question for future work whether or not the construction can be optimised to obtain ``nice'' disjunctive formulas. In particular, the move from modal automata to parity formulas potentially adds a large number of unnecessary disjuncts. Obtaining a nicer formula could be relevant for computing uniform interpolants. 
Another nagging question is the exact repercussion of our work for satisfiability checking. While satisfiability checking for disjunctive formulas is linear in subformula size, our formulas are measured in closure size, which is potentially an exponential smaller. It has to be checked whether one can use our result for ExpTime satisfiability checking when the input formula is measured in closure size.

\bibliographystyle{eptcs}
\bibliography{mu.bib,extra.bib}




\end{document}